\documentclass{tlp}
\usepackage{aopmath}

\usepackage{epic}
\usepackage{eepic}

\def\rar{\rightarrow}
\def\lrar{\leftrightarrow}
\def\beq{\begin{equation}}
\def\eeq#1{\label{#1}\end{equation}}
\def\ba{\begin{array}}
\def\ea{\end{array}}
\def\i#1{\hbox{\it #1\/}}
\def\is#1{{\hbox{\scriptsize {\it #1\/}}}}

\def\eq{\i{Eq}}
\def\eqs{\is{Eq}}

\def\sm{\hbox{\rm SM}}
\def\S1eq{\sim_{S_1}}

\def\i#1{\hbox{\itshape #1\/}}

\newtheorem{thm}{Theorem}

\title[Relational Theories with Null Values]{\bf Relational Theories
with Null Values\\ and Non-Herbrand Stable Models}
 \author[ Vladimir Lifschitz, Karl Pichotta, and  Fangkai Yang]
 {Vladimir Lifschitz, Karl Pichotta, and  Fangkai Yang\\
 Department of Computer Science\\
 University of Texas at Austin\\
 \{vl,pichotta,fkyang\}@cs.utexas.edu
}
\begin{document}

\date{}
\maketitle

\begin{abstract}
Generalized relational theories with null values in the sense of Reiter
are first-order theories that provide a semantics for relational databases
with incomplete information.  In this paper we show that any such theory
can be turned into an equivalent logic program, so that models of the theory
can be generated
using computational methods of answer set programming.  As a step towards this
goal, we develop a general method for calculating stable models under the
domain closure assumption but without the unique name assumption.
\end{abstract}

\section{Introduction}

We re-examine here some of the problems discussed in two important
papers on the semantics of null values that were published many years ago.
The first of them is Ray Reiter's paper ``Towards a
logical reconstruction of relational database theory'' \cite{rei84}.
Generalized relational theories with null values in the sense of Reiter
are first-order theories that provide a semantics for relational databases
with incomplete information.  The incompleteness can be of two kinds.  One is
represented by inclusive disjunction; for instance, the formula
\beq
\i{SUPPLIES}(\i{Foo}, p_1) \vee \i{SUPPLIES}(\i{Foo}, p_3)
\eeq{dis}
says: \i{Foo} supplies~$p_1$ or~$p_3$, maybe both.  The other is represented by
null values; by writing
\beq
\i{SUPPLIER}(\omega),\ \i{SUPPLIES}(\omega,p_3),
\eeq{nvexample}
where $\omega$ is a null value, we express that some supplier, which may or
may not already be in the database, supplies~$p_3$.

The second paper, by Bonnie Traylor and Michael Gelfond, is entitled
``Representing null values in logic programming'' \cite{tra94}.  The authors
define, among other things, the ``logic programming counterpart'' of a
generalized relational theory with null values---a logic program whose
meaning under the answer set semantics is similar to the meaning of the
theory under the standard semantics of first-order logic.

We propose here an alternative approach to turning Reiter's theories into logic
programs, which represents the meaning of the
theory more closely than the translation from \cite{tra94}.  We show also
how these logic programs can be executed using computational
methods of answer set programming (ASP) \cite{mar99,nie99,lif08}---for
instance, by running the answer set solver {\sc clingo}.\footnote{\tt
http://potassco.sourceforge.net \label{ft}}

The difference between null values and other object constants emphasized in
Reiter's semantics is that null values are exempt from the unique name
assumption: a null value may represent an object that has a name in the
database, and two different null values may represent the same object.  This
fact leads us to the general problem of using answer set solvers
for calculating the stable models that satisfy the domain closure assumption
but may not satisfy the unique name assumption.  Such models are allowed in
some versions of the stable model semantics \cite{fer07a,fer09}, just as they
are allowed in the definition of circumscription \cite{mcc80,mcc86}.
But existing answer set solvers do not deal with stable models of this kind
directly.  To take a simple example, the formula
\beq
P(a)\lor P(b)
\eeq{easy}
has minimal models of three kinds: in some of them,~$P(a)$ is true, and~$P(b)$
is false; in others,~$P(a)$ is false, and~$P(b)$ is true; finally, there are
minimal
models in which both~$P(a)$ and~$P(b)$ are true, along with the formula $a=b$.
We will see how syntactic expressions describing these three possibilities
can be generated by an answer set solver.  Our method is applicable, in
particular, to logic programs representing relational theories with null values.


The word ``generalized'' in Reiter's terminology indicates the possibility of
including disjunctive information, and in the rest of the paper it will be
omitted.

\section{Relational Theories without Null Values} 

\subsection{Review of Reiter's Semantics of Relational Theories} 

We begin with a signature that consists of finitely many object and predicate
constants.  A {\sl positive ground clause} is a formula of the form
$A_1\lor\cdots\lor A_r$ ($r\geq 1$), where each~$A_i$ is a ground atomic
formula whose predicate is distinct from the equality symbol.  For
instance,~(\ref{dis}) is a positive ground clause.  For any finite
set~$\Delta$ of positive ground clauses, the corresponding {\sl relational
theory} $T_\Delta$ is the set consisting of the following sentences:
\begin{itemize}
\item
 the {\sl domain closure axiom} \i{DCA\/}:
$$
\forall x \bigvee_a x=a
$$
  where the disjunction extends over all object constants~$a$;
\item
  the {\sl unique name axioms} $a\neq b$ for all pairs of distinct object
constants~$a$,~$b$;
  \item
  the clauses $\Delta$;
  \item
  for each predicate constant~$P$, the {\sl completion axiom}
\beq
  \forall {\bf x} \left\lbrack P({\bf x}) \rightarrow
           \bigvee_{{\bf a}\in W_P} {\bf x} = {\bf a}\right\rbrack
\eeq{comp}
  where {\bf x} is a tuple of distinct object variables, and $W_P$
  is the set of all tuples~{\bf a} of object constants such that
  $P({\bf a})$ belongs to a clause from~$\Delta$.\footnote{The equality between
two tuples of terms of the same length, such as ${\bf x}={\bf a}$, stands
for the conjunction of the equalities between the corresponding members of the
tuples. We do not include equality axioms from \cite{rei84} because we assume
here the version of the semantics of first-order formulas that treats
equality as identity (see, for instance, \cite[Section~1.2.2]{lif08b}).}
\end{itemize}

In view of the domain closure axiom \i{DCA} and the unique name axioms, any
model
of~$T_\Delta$ is isomorphic to a Herbrand model.\footnote{Recall that in
the absence of function constants of arity $>0$ a {\sl Herbrand
interpretation} is an interpretation such that (i)~its universe is the set of
all object constants, and (ii)~each object constant is interpreted as itself.
A Herbrand interpretation can be identified with the set of all ground atomic
formulas that are true in it and whose predicate is distinct from the
equality symbol.}  Consequently, in the discussion of models of ~$T_\Delta$
we can restrict attention to Herbrand models.

Consider, for instance, Example 4.1 from \cite{rei84}.  Its signature includes
the object constants
$$p_1,\ p_2,\ p_3,\ \i{Acme},\ \i{Foo},$$
the unary predicate constants
$$\i{PART},\ \i{SUPPLIER},$$
and the binary predicate constants
$$\i{SUPPLIES},\ \i{SUBPART}.$$
The set~$\Delta$ describes the following supplier and parts world:
\begin{center}
\begin{tabular}{cccc}
{\it PART} & {\it SUPPLIER} & {\it SUPPLIES} & {\it SUBPART} \\
$p_1$ & \i{Acme} & $\langle$\i{Acme}$\,\,\,\,p_1\rangle$ & $\langle p_1\,\,\,\,p_2 \rangle$ \\
$p_2$ & \i{Foo} & $\langle$\i{Foo}$\,\,\,\,p_2\rangle$ &\\
$p_3$
\end{tabular}
\end{center}
In other words, it includes the corresponding atomic formulas:
\beq
\i{PART}(p_1),\ \i{PART}(p_2),\ \dots,\ \i{SUBPART}(p_1,p_2).
\eeq{common}
In addition, $\Delta$ includes clause~(\ref{dis}).

The completion axioms in this example are
$$
\ba l
\forall x(\i{PART}(x) \rar x=p_1\lor x=p_2\lor x=p_3),\\
\forall x(\i{SUPPLIER}(x) \rar x=\i{Acme}\lor x=\i{Foo}),\\
\forall xy(\i{SUPPLIES}(x,y) \rar
   (x=\i{Acme}\land y= p_1) \\
   \qquad\qquad\qquad\qquad\qquad\quad \lor\; (x=\i{Foo} \land y = p_2)\\
   \qquad\qquad\qquad\qquad\qquad\quad \lor\; (x=\i{Foo} \land y = p_1)\\
     \qquad\qquad\qquad\qquad\qquad\quad \lor\; (x=\i{Foo} \land y = p_3)),\\
\forall xy(\i{SUBPART}(x,y) \rar (x=p_1 \land y = p_2)).
\ea
$$
Theory~$T_\Delta$ has 3 Herbrand models:
$$\ba l
I_1=I\cup\{\i{SUPPLIES}(\i{Foo}, p_1)\},\\
I_2=I\cup\{\i{SUPPLIES}(\i{Foo}, p_3)\},\\
I_3=I\cup\{\i{SUPPLIES}(\i{Foo}, p_1),\i{SUPPLIES}(\i{Foo}, p_3)\},
\ea$$
where~$I$ is the set of atomic formulas~(\ref{common}).

Note that~$I_3$ is not a minimal model of~$T_\Delta$: both~$I_1$ and~$I_2$ are
proper subsets of~$I_3$.  In the presence of disjunction, Reiter's completion
axioms~(\ref{comp}) guarantee only a weak form of minimality.  A similar
condition is used in the definition of the possible model semantics of
disjunctive logic programs \cite{sak94}.

\subsection{Representing Relational Theories by Logic Programs}
\label{ssec:trans1}

For any set~$\Delta$ of positive ground clauses, by $\Pi_\Delta$ we denote
the set of rules
\beq
1\{A_1,\dots,A_r\}
\eeq{pdrule}
for all clauses $A_1\lor\cdots\lor A_r$ from~$\Delta$.  Recall that this is an
expression in the input language of {\sc clingo}\footnote{Such expressions,
``cardinality constraints,''
appeared originally as part of the input language of the grounder
{\sc lparse} ({\tt http://www.tcs.hut.fi/Software/smodels/lparse.ps.gz}).}
that allows us to decide arbitrarily whether or not to include the atomic
formulas~$A_1,\dots,A_r$ in the answer set as long as at least one of them is
included.

The translation $1\{A\}$ of a unit clause $A$ is strongly equivalent
\cite{lif01,lif07a} to the fact~$A$.  Using this simplification we can say, for
instance, that the logic program representing the example above
consists of the facts~(\ref{common}) and the rule
$$1\{\i{SUPPLIES}(\i{Foo}, p_1),\i{SUPPLIES}(\i{Foo}, p_3)\}.$$
Furthermore, this program can be made more compact using the {\sc clingo}
conventions that allow us to use semicolons to merge a group of facts into one
expression:
\begin{verbatim}
part(p1;p2;p3).
supplier(acme;foo).
supplies(acme,p1;;foo,p2).
subpart(p1,p2).
1{supplies(foo,p1),supplies(foo,p3)}.
\end{verbatim}
Given this input, {\sc clingo} returns 3 answer sets:
\begin{verbatim}
Answer: 1
part(p1) part(p2) part(p3) supplier(acme) 
supplier(foo) supplies(foo,p2) supplies(acme,p1) 
subpart(p1,p2) supplies(foo,p3)
Answer: 2
part(p1) part(p2) part(p3) supplier(acme) 
supplier(foo) supplies(foo,p2) supplies(acme,p1) 
subpart(p1,p2) supplies(foo,p1)
Answer: 3
part(p1) part(p2) part(p3) supplier(acme) 
supplier(foo)supplies(foo,p2) supplies(acme,p1) 
subpart(p1,p2) supplies(foo,p1) supplies(foo,p3)
\end{verbatim}

These answer sets are identical to the Herbrand models of the corresponding
relational theory.  This is an instance of a general theorem
that expresses the correctness of our translation:

\begin{thm}\label{thm1}
For any set~$\Delta$ of positive ground clauses, a Herbrand
interpretation~$I$ is a model of~$T_\Delta$ iff $I$ is an answer set
of~$\Pi_\Delta$.
\end{thm}

Proofs of theorems, including a combined proof of Theorems~1 and~2, can be
found at the end of the paper.

\section{Null Values} 

\subsection{Review of Reiter's Semantics of Null Values} 

We turn now to a more general framework.  As before, the
underlying signature is assumed to consist of finitely many object and
predicate constants.  We assume that the object constants are classified into
two groups, {\sl the database constants} and {\sl the null values}.  About a
unique name axiom $a\neq b$ we say that it is {\sl required} if both~$a$
and~$b$ are database constants, and that it is {\sl optional} otherwise.
As before, $\Delta$ stands for a finite
set of positive ground clauses.  Let~$\Sigma$ be a set of optional unique
name axioms.  The {\sl relational theory with null values} $T_{\Delta,\Sigma}$
is the set of sentences obtained from $T_\Delta$ by removing all optional
unique name axioms that do not belong to~$\Sigma$.  In other words,
$T_{\Delta,\Sigma}$ consists of
\begin{itemize}
\item
 the domain closure axiom \i{DCA},
\item
all required unique name axioms,
\item
the optional unique name axioms from~$\Sigma$,
\item
  the clauses $\Delta$;
\item
  the completion axioms~(\ref{comp}).
\end{itemize}

Consider, for instance, the modification of our example in which
\begin{itemize}
\item
the object constant~$\omega$ is added to signature as the only null value,
\item
clause~(\ref{dis}) is replaced in~$\Delta$ with clauses~(\ref{nvexample}),
and
\item
$\Sigma=\{\omega\neq p_1,\omega\neq p_2,\omega\neq p_3\}$.
\end{itemize}
Thus $\omega$ is assumed to be a supplier that supplies part~$p_3$; it may
be identical to one of the suppliers \i{Acme}, \i{Foo} or may be different
from both of them, and it is certainly different from~$p_1$,~$p_2$,~$p_3$.
The completion axioms in this example are
$$
\ba l
\forall x(\i{PART}(x) \rar x=p_1\lor x=p_2\lor x=p_3),\\
\forall x(\i{SUPPLIER}(x) \rar x=\i{Acme}\lor x=\i{Foo}\lor x=\omega),\\
\forall xy(\i{SUPPLIES}(x,y) \rar
   (x=\i{Acme}\land y= p_1) \\
   \qquad\qquad\qquad\qquad\qquad\quad \lor\;(x=\i{Foo} \land y = p_2)\\
   \qquad\qquad\qquad\qquad\qquad\quad \lor\;(x=\omega \land y = p_3)),\\
\forall xy(\i{SUBPART}(x,y) \rar (x=p_1 \land y = p_2)).
\ea
$$
The set of unique name axioms of~$T_{\Delta,\Sigma}$ includes neither
\hbox{$\omega\neq\i{Acme}$} nor \hbox{$\omega\neq\i{Foo}$}.  Accordingly, this
theory has models of different kinds: some of them satisfy
\hbox{$\omega=\i{Acme}$}; some satisfy \hbox{$\omega=\i{Foo}$}; in some models,
both equalities are false.  We will later return to this example to give
a complete description of its models.

\subsection{Representing Theories with Null Values by Logic Programs}

In Section~\ref{ssec:trans1} we saw how Reiter's semantics of disjunctive
databases can be reformulated in terms of stable models.  Our next goal is to
do the same for databases with null values.

Since the axiom set~$T_{\Delta,\Sigma}$ may not include some of the optional
unique name axioms, it may have models that are not isomorphic to any Herbrand
model.  For this reason, the problem of relating~$T_{\Delta,\Sigma}$ to logic
programs becomes easier if we start with a semantics of logic programs that
is not restricted to Herbrand models.

A version of the stable model semantics that covers non-Herbrand models is
described in \cite[Section~2]{fer09}.\footnote{Other possible approaches to
the semantics of logic programs that are
not limited to Herbrand models use program completion \cite{cla78}
without Clark's equality axioms and the logic of nonmonotone inductive
definitions \cite{den08}.}  That
paper deals with models of a first-order
sentence and defines under what conditions such a model is considered stable
relative to a subset~{\bf p} of the predicate constants of the underlying
signature.  The predicates from~{\bf p} are called ``intensional.''  Unless
stated otherwise, we will assume that {\bf p} consists of all predicate
constants of the underlying signature, so that every predicate constant (other
than equality) is considered intensional.
When this definition of a stable model is applied to a logic program, each rule
of the program is viewed as shorthand for a first-order sentence, and the
program is identified with the conjunction of these sentences.  For instance,
rule~(\ref{pdrule}) can be viewed as shorthand for the formula
$$\bigwedge_{i=1}^r(A_i\lor \neg A_i)\wedge\bigvee_{i=1}^r A_i.$$
(The first conjunctive term says, ``choose the truth value of each~$A_i$
arbitrarily''; the second term requires that at least one of these atoms be
made true.)

The paper referenced above defines a syntactic transformation $\sm_{\bf p}$ that
turns a first-order sentence~$F$ into a conjunction
$$F\land\,\cdots$$
where the dots stand for a second-order sentence (the ``stability condition'').
The stable models of~$F$ are defined as arbitrary models (in the sense of
second-order logic) of $\sm_{\bf p}[F]$.

From this perspective, Theorem~\ref{thm1} asserts that a Herbrand
interpretation is a model of~$T_\Delta$ iff it is a model
of~$\sm_{\bf p}[\Pi_\Delta]$, where~{\bf p} is the set of all predicate
constants of the underlying signature.

By~$\Pi_{\Delta,\Sigma}$ we denote the conjunction of~$\Pi_\Delta$ with \i{DCA}
and with all unique name axioms from
$T_{\Delta,\Sigma}$ (that is to say, with all unique name axioms except for
the optional axioms that do not belong to~$\Sigma$).  The following theorem
expresses the soundness of this translation:

\begin{thm}\label{thm2}
For any set~$\Delta$ of positive ground clauses and any set~$\Sigma$ of optional
unique name axioms, $T_{\Delta,\Sigma}$ is equivalent to
$\sm_{\bf p}[\Pi_{\Delta,\Sigma}]$, where~{\bf p} is the set of all predicate
constants.
\end{thm}

In other words, an interpretation~$I$ is a model
of $T_{\Delta,\Sigma}$ iff~$I$ is a stable model of~$\Pi_{\Delta,\Sigma}$.

One useful property of the operator $\sm_{\bf p}$ is that
$$\sm_{\bf p}[F\land G]\hbox{ is equivalent to }\sm_{\bf p}[F]\land G$$
whenever~$G$ does not contain intensional predicates (that is, predicate
constants from~{\bf p}).\footnote{See \cite[Section~5.1]{fer09}.}  For
instance, let~$\Pi^{-}_{\Delta,\Sigma}$ be the conjunction of~$\Pi_\Delta$ with the
unique name axioms from $T_{\Delta,\Sigma}$; then $\Pi_{\Delta,\Sigma}$ is
$\Pi^{-}_{\Delta,\Sigma}\land\i{DCA}$.  Since \i{DCA} does not contain intensional
predicates (recall that all atomic parts of \i{DCA} are equalities),
$\sm_{\bf p}[\Pi_{\Delta,\Sigma}]$ is equivalent to
$\sm_{\bf p}[\Pi^{-}_{\Delta,\Sigma}]\land\i{DCA}$.  The assertion of
Theorem~\ref{thm2} can be reformulated as follows: an interpretation~$I$ is a
model of $T_{\Delta,\Sigma}$ iff~$I$ is a stable model of $\Pi^{-}_{\Delta,\Sigma}$
that satisfies \i{DCA}.

As we have seen, the translation $\Pi_\Delta$ makes it
possible to generate models of~$T_\Delta$ using an answer set solver.
Unfortunately, the translation $\Pi_{\Delta,\Sigma}$ does not do the same for
relational theories with null values.  In the presence of null values we are
interested in non-Herbrand models (for instance, in the models of the theory
from the example above that satisfy $\omega=\i{Acme}$), but answer set
solvers are designed to generate Herbrand stable models only.  There is also a
more basic question: a Herbrand interpretation can be viewed as a set of
ground atomic formulas, but how will we describe non-Herbrand models by
syntactic expressions?  These questions are addressed in the next section.

\section{Calculating General Stable Models} 

\subsection{Diagrams} 

Consider a signature~$\sigma$ consisting of finitely many object and predicate
constants.  By $\i{HB}_\sigma$ we denote the Herbrand base of~$\sigma$, that is,
the set of its ground atomic formulas whose predicate is distinct from the
equality symbol.  By $\i{EHB}_\sigma$ (``extended'' Herbrand base) we denote the
set of all ground atomic formulas, including equalities between object
constants.  For any interpretation~$I$ of~$\sigma$ satisfying the domain
closure axiom ({\sl \i{DCA}-interpretation}, for short), by~$D(I)$ we will
denote the set of the formulas from~$\i{EHB}_\sigma$ that are true in~$I$.
This set
will be called the {\sl diagram} of~$I$.\footnote{This is essentially the
``positive diagram'' of~$I$, as this term is used in model
theory \cite[Section~2.1]{rob63}, for the special case of
\i{DCA\/}-interpretations.}

If a subset~$X$ of~$\i{EHB}_\sigma$ is the diagram of a \i{DCA\/}-interpretation
then it is clear that
\begin{itemize}
\item
the set of equalities in~$X$ is closed under reflexivity (it includes $a=a$ for
every object constant~$a$), symmetry (includes~$b=a$
whenever it includes~$a=b$), and transitivity (includes~$a=c$ whenever it
includes~$a=b$ and~$b=c$), and
\item
$X$ is closed under substitution: it includes $P(b_1,\dots,b_n)$ whenever it
includes $P(a_1,\dots,a_n)$, $a_1=b_1,\dots,a_n=b_n$.
\end{itemize}
The converse holds also:
\begin{thm}\label{thm3}
If a subset~$X$ of~$\i{EHB}_\sigma$ is closed under substitution, and the set of
equalities in~$X$ is closed under reflexivity, symmetry, and transitivity,
then there exists a \i{DCA}-interpretation~$I$ such that $D(I)=X$.
Furthermore, this interpretation is unique up to isomorphism.
\end{thm}

Since relational theories with null values include the domain closure
assumption, Theorem~\ref{thm3} shows that their models can be completely
characterized by diagrams.  In the example above,
the theory has~3 non-isomorphic models~$J_1$,~$J_2$,~$J_3$.  The diagram
of~$J_1$ consists of the formulas~(\ref{common}),~(\ref{nvexample}), and
$a=a$ for all object constants~$a$.  The diagrams of the other two are given
by the formulas
$$
\ba l
J_2=J_1\cup\{\i{SUPPLIES}(\i{Acme}, p_3),\i{SUPPLIES}(\omega,p_1),
\omega=\i{Acme},\i{Acme}=\omega\},\\
J_3=J_1\cup\{\i{SUPPLIES}(\i{Foo}, p_3),\i{SUPPLIES}(\omega,p_2),
\omega=\i{Foo},\i{Foo}=\omega\}.
\ea
$$

\subsection{Reducing Stable DCA-Models to Herbrand Stable Models}

The problem that we are interested in can be now stated as follows: Given
a first-order sentence~$F$, we would like to construct a first-order
sentence~$F'$ such that the diagrams of all \i{DCA\/}-interpretations
satisfying~$\sm_{\bf p}[F]$ can be easily extracted from the Herbrand
interpretations satisfying~$\sm_{\bf p}[F']$.  We say ``can be easily
extracted from'' rather than ``are identical to'' because diagrams include
equalities between object constants, and Herbrand models do not; occurrences
of equality in~$F$ will have to be replaced in~$F'$ by another symbol.
Our goal, in other words, is to define~$F'$ in such a way that diagrams of
the stable \i{DCA\/}-models of~$F$ will be nearly identical to Herbrand stable
models of~$F'$.

The examples of~$F$ that we are specifically interested in are the
formulas $\Pi^{-}_{\Delta,\Sigma}$, because stable \i{DCA\/}-models of that
sentence are identical to models of $T_{\Delta,\Sigma}$.
As a simpler example, consider
formula~(\ref{easy}). It has 3 minimal \i{DCA\/}-models, with the diagrams
\beq
\ba l
K_1=\{P(a), a=a, b=b\},\\
K_2=\{P(b), a=a, b=b\},\\
K_3=\{P(a), P(b), a=a, b=b, a=b, b=a\}.
\ea
\eeq{mineasy}
Our translation $F\mapsto F'$ will allow us to construct these diagrams
using ASP.

The solution described below uses the binary predicate constant~$\eq$, which is
assumed not to belong to~$\sigma$.  For any first-order formula~$F$ of the
signature~$\sigma$, $F^{=}_{\eqs}$ stands for the formula of the signature
$\sigma\cup\{\eq\}$ obtained from~$F$ by replacing each subformula of the form
$t_1=t_2$ with $\eq(t_1,t_2)$.  (Here~$t_1$,~$t_2$ are terms, that is, object
constants or object variables.)  The notation $X^{=}_{\eqs}$, where~$X$ is a set of
formulas of the signature~$\sigma$, is understood in a similar way.
By $E_{\sigma}$ we denote the conjunction of the logically valid sentences
$$
\ba c
\forall x (x=x).\\
\forall xy(x=y\rar y=x),\\
\forall xyz(x=y\land y=z\rar x=z),
\ea
$$
and
$$
\forall {\bf xy}(P({\bf x})\land {\bf x}={\bf y}\rar P({\bf y}))
$$
for all predicate constants~$P$ from~$\sigma$, where {\bf x}, {\bf y} are
disjoint tuples of distinct variables.

In the statement of the theorem below,~$F$ is an arbitrary sentence of the
signature~$\sigma$, and~{\bf p} stands for the set of all predicate constants
of~$\sigma$.

\begin{thm}\label{thm4}
For any \i{DCA}-interpretation~$I$ of the signature~$\sigma$ that satisfies
$\sm_{\bf p}[F]$, the Herbrand interpretation $D(I)^{=}_{\eqs}$ of the
signature~$\sigma\cup\{\eq\}$ satisfies
\beq
\sm_{\bf p}[(F\land E_{\sigma})^{=}_{\eqs}].
\eeq{res}
Conversely, any Herbrand model of this formula is $D(I)^{=}_{\eqs}$ for some
\i{DCA}-interpreta\-tion~$I$ of~$\sigma$ satisfying $\sm_{\bf p}[F]$.
\end{thm}

In other words, the transformation $I\mapsto D(I)^{=}_{\eqs}$ maps the class of
stable \i{DCA}-models of~$F$ onto the set of Herbrand stable models of
$(F\land E_{\sigma})^{=}_{\eqs}$.
The second part of Theorem~\ref{thm3} shows that this transformation
is one-to-one up to isomorphism.

By Theorem~2 from \cite{fer09}, formula~(\ref{res}) is equivalent to
\beq
\sm_{{\bf p},\eqs}[(F\land E_{\sigma})^{=}_{\eqs}\land\forall xy(\eq(x,y)\lor\neg \eq(x,y))].
\eeq{res1}
The advantage of this reformulation is that it treats all predicate
constants of the signature $\sigma\cup\{\eq\}$ as intensional.  This is
essential for our purposes,
because existing answer set solvers calculate Herbrand stable models
under the assumption that all predicate constants occurring in the program
(except for ``predefined predicates'') are intensional.

For example, the diagrams~(\ref{mineasy}) of the minimal {\sl DCA}-models
of~(\ref{easy}) are identical, modulo replacing $=$ with~$\eq$, to the
Herbrand stable models of the conjunction of the formulas~(\ref{easy}),
\beq
\ba c
\forall x \eq(x,x),\\
\forall xy(\eq(x,y)\rar \eq(y,x)),\\
\forall xyz(\eq(x,y)\land \eq(y,z)\rar \eq(x,z)),\\
\ea
\eeq{eq0}
$$
\forall xy(P(x)\land \eq(x,y)\rar P(y)),
$$
and
\beq
\forall xy(\eq(x,y)\lor\neg \eq(x,y)).
\eeq{choice}
In logic programming syntax, this conjunction can be written as
\begin{verbatim}
p(a)|p(b).
eq(X,X).
eq(X,Y) :- eq(Y,X).
eq(X,Z) :- eq(X,Y), eq(Y,Z).
p(Y) :- p(X), eq(X,Y).
{eq(X,Y)}.
\end{verbatim}
To make this program safe\footnote{Safety is a syntactic condition
required for ``intelligent instantiation''---part of
the process of generating answer sets.  In the program above, the rules
{\tt eq(X,X)} and {\tt \{eq(X,Y)\}} are unsafe.} we need to specify that the
only possible values of the variables {\tt X} and {\tt Y} are {\tt a} and
{\tt b}.  This can be accomplished by including the lines
\begin{verbatim}
u(a;b).
#domain u(X). #domain u(Y).
#hide u/1.
\end{verbatim}
(The auxiliary predicate symbol {\tt u} describes the ``universe'' of the
program.)  Now the program can be grounded by {\sc gringo}, and its Herbrand
stable models can be generated by {\sc claspD}.\footnote{{\sc gringo} and
{\sc claspD} are ``relatives''  of  {\sc clingo}; see Footnote ($^{\ref{ft}}$)
for a reference.
{\sc clingo} itself cannot be used in this case because the program is
disjunctive. \hbox{\sc cmodels}
({\tt http: //www.cs.utexas.edu/users/tag/cmodels.html}) would do as well.
Using the solver {\sc dlv} ({\tt http://www.dlvsystem.com}) will become an
option too after eliminating choice rules in favor of disjunctive rules with
auxiliary
predicates.  We are grateful to Yuliya Lierler for helping us
identify the software capable of executing this program.}  The output
\begin{verbatim}
Answer: 1
eq(b,b) eq(a,a) p(b)
Answer: 2
eq(b,b) eq(a,a) p(a)
Answer: 3
eq(b,b) eq(a,a) eq(b,a) eq(a,b) p(a) p(b)
\end{verbatim}
is essentially identical to the list~(\ref{mineasy}) of minimal models, as
could be expected on the basis of Theorem~\ref{thm4}.

The Python script {\sc nonH.py} (for ``non-Herbrand'') is a preprocessor
that turns a program~$F$ of a signature~$\sigma$ without function symbols of
arity $>0$, written in the input language of {\sc gringo}, into the program
$$(F\land E_{\sigma})^{=}_{\eqs}\land\forall xy(\eq(x,y)\lor\neg \eq(x,y)),$$
written in the language of {\sc gringo} also.  Thus the Herbrand stable models
of the
output of {\sc nonH.py} are the diagrams of the stable \i{DCA}-models of the
input (with equality replaced by $\eq$).  As in the example above, a ``universe''
predicate is used to ensure that whenever the input of {\sc nonH.py} is safe,
the output is safe also.  The diagrams of the minimal \i{DCA}-models of
formula~(\ref{easy}) can be generated by saving that formula, in the form
\begin{verbatim}
p(a)|p(b).
\end{verbatim}
in a file, say {\tt disjunction.lp}, and then executing the command
\begin{verbatim}
% nonH.py disjunction.lp | gringo | claspD 0
\end{verbatim}
(the {\sc claspD} option 0 instructs it to generate all answer sets, not
one).  The script can be downloaded from {\tt
http://www.cs.utexas.edu/users/fkyang/nonH/}.

\subsection{Calculating Models of a Relational Theory with Null Values}

The method applied above to the disjunction $P(a)\lor P(b)$ can be applied also
to the formula $\Pi^{-}_{\Delta,\Sigma}$. Stable \i{DCA}-models of this formula
can be generated using
{\sc clingo} with the preprocessor {\sc nonH.py}.  The preprocessor has two
options that can be useful here.  The command line

\medskip\noindent
{\tt \% nonH.py} $<$filename$>$ {\tt -una} $<$list of constants$>$

\medskip\noindent
instructs the preprocessor to conjoin its output with the unique name axioms
$a\neq b$ for all pairs $a$, $b$ of distinct object constants from the given
list.  The command line

\medskip\noindent
{\tt \% nonH.py} $<$filename$>$ {\tt -no-una} $<$list of constants$>$

\medskip\noindent
adds the unique name axioms $a\neq b$ for all pairs $a$, $b$ of distinct
object constants such that at least one of them does not occur in the given
list.  The diagrams of models of our relational theory with null values can
be generated by saving the rules
\begin{verbatim}
part(p1;p2;p3).
supplier(acme;foo;omega).
supplies(acme,p1;;foo,p2;;omega,p3).
subpart(p1,p2).
:- omega==p1.
:- omega==p2.
:- omega==p3.
\end{verbatim}
in a file, say {\tt db.lp}, and then executing the command
\begin{verbatim}
% nonH.py db.lp -no-una omega | clingo 0
\end{verbatim}
The following output will be produced:
\begin{verbatim}
Answer: 1
part(p1) part(p3) part(p2) supplier(acme) supplier(omega) supplier(foo)
supplies(omega,p3) supplies(foo,p2) supplies(acme,p1) subpart(p1,p2)
eq(omega,omega) eq(foo,foo) eq(acme,acme) eq(p3,p3) eq(p2,p2) eq(p1,p1)
eq(omega,foo) eq(foo,omega) supplies(omega,p2) supplies(foo,p3)
Answer: 2
part(p1) part(p3) part(p2) supplier(acme) supplier(omega) supplier(foo)
supplies(omega,p3) supplies(foo,p2) supplies(acme,p1) subpart(p1,p2)
eq(omega,omega) eq(foo,foo) eq(acme,acme) eq(p3,p3) eq(p2,p2) eq(p1,p1)
Answer: 3
part(p1) part(p3) part(p2) supplier(acme) supplier(omega) supplier(foo)
supplies(omega,p3) supplies(foo,p2) supplies(acme,p1) subpart(p1,p2)
eq(omega,omega) eq(foo,foo) eq(acme,acme) eq(p3,p3) eq(p2,p2) eq(p1,p1)
eq(omega,acme) eq(acme,omega) supplies(acme,p3) supplies(omega,p1)
\end{verbatim}
It is essentially identical to the set of diagrams $J_1,J_2,J_3$.

\subsection{Comparison with the Traylor---Gelfond Translation}

The approach to encoding relational theories with null values by logic
programs proposed in \cite{tra94} does not have the property established
for~$\Pi_{\Delta,\Sigma}$ in Theorem~\ref{thm2}: generally, there is no
1--1 correspondence between the models of~$T_{\Delta,\Sigma}$ and the answer
sets of the Traylor---Gelfond translation.  For instance, the logic programming
counterpart of our main example in the sense of
\cite{tra94} has~2 answer sets, not~3.
It uses strong (classical) negation \cite{gel91b}, and its
answer sets are incomplete sets of literals.  One of them, for instance,
includes $\i{SUPPLIES}(\i{Foo}, p_1)$ but does not include either of the two
complementary literals $\i{SUPPLIES}(\i{Foo}, p_3)$,
$\neg\i{SUPPLIES}(\i{Foo}, p_3)$.  This is how the program expresses the
possibility of $p_3$ being supplied by \i{Foo}, along with~$p_1$.
The result of \cite{tra94} describes the relation of~$T_{\Delta,\Sigma}$ to the
intersection of the answer sets of its logic programming counterpart, not to
the individual answer sets.

Logic programming counterparts in the sense of \cite{tra94}, like our
programs~$\Pi_{\Delta,\Sigma}$, can be turned into executable ASP code.  The
reason why that was not done in that paper is simply that the paper was written
too early---the first answer set solver appeared on the scene two years after
its publication \cite{nie96}.

\section{Proofs of Theorems} 

\subsection{Proofs of Theorems~\ref{thm1} and~\ref{thm2}}

\begin{lemma}\label{lem1}
For any finite set~$\Delta$ of positive ground clauses,
formula $\sm_{\bf p}[\Pi_\Delta]$ is equivalent to the conjunction of the
clauses~$\Delta$ and the completion axioms~(\ref{comp}).
\end{lemma}

\begin{proof}
Let~$C$ be the conjunction of the formulas
\beq
P({\bf a})\lor \neg P({\bf a})
\eeq{pl1}
for all atomic formulas $P({\bf a})$ occurring in~$\Delta$.  It is clear
that~$\Pi_\Delta$ is strongly equivalent\footnote{See \cite[Section~5]{fer09}.}
to the conjunction of~$C$ with the formulas
\beq
\neg\bigwedge_{i=1}^r\neg A_i
\eeq{co}
for all clauses~$A_1\lor\cdots\lor A_r$ from~$\Delta$.  According to
Theorem~3 from
~\cite{fer09}, it follows that $\sm_{\bf p}[\Pi_\Delta]$ is
equivalent to the conjunction of $\sm_{\bf p}[C]$ with formulas~(\ref{co}).
Furthermore,~(\ref{pl1}) is strongly equivalent to
$\neg\neg P({\bf a})\rar P({\bf a})$.  Consequently~$C$ is strongly
equivalent to the conjunction of the formulas
$$\forall {\bf x} \left\lbrack
   \bigvee_{{\bf a}\in W_P}\!\!(\neg\neg P({\bf x})\land {\bf x} = {\bf a})\;
      \rar\; P({\bf x})\right\rbrack$$
for all predicate constants~$P$.  By Theorem~11 from~\cite{fer09}, it follows
that $\sm_{\bf p}[C]$ is equivalent to
\beq
\forall {\bf x} \left\lbrack P({\bf x})\;\lrar\;
   \bigvee_{{\bf a}\in W_P}\!\!(\neg\neg P({\bf x})\land {\bf x} = {\bf a})
      \right\rbrack.
\eeq{comp1}
It remains to observe that~(\ref{co}) is equivalent to~$A_1\lor\cdots\lor A_r$,
and that~(\ref{comp1}) is equivalent to~(\ref{comp}).
\end{proof}

\medskip\noindent{\bf Theorem~\ref{thm1}.} {\it
For any set~$\Delta$ of positive ground clauses, a Herbrand
interpretation~$I$ is a model of~$T_\Delta$ iff $I$ is an answer set
of~$\Pi_\Delta$.}

\medskip\begin{proof}
A Herbrand interpretation is a model of~$T_\Delta$ iff it satisfies
the clauses~$\Delta$ and the completion axioms~(\ref{comp}).  On the other
hand, a Herbrand interpretation is an answer set of~$\Pi_\Delta$ iff it
satisfies $\sm_{\bf p}[\Pi_\Delta]$.  Consequently the assertion of the
theorem follows from Lemma~\ref{lem1}.
\end{proof}

\medskip\noindent{\bf Theorem~\ref{thm2}.} {\it
For any set~$\Delta$ of positive ground clauses and any set~$\Sigma$ of optional
unique name axioms, $T_{\Delta,\Sigma}$ is equivalent to
$\sm_{\bf p}[\Pi_{\Delta,\Sigma}]$, where~{\bf p} is the set of all predicate
constants.
}

\medskip\begin{proof}
Recall that $\Pi_{\Delta,\Sigma}$ is~$\Pi_\Delta\land\i{DCA}\land U$, where~$U$
is the conjunction of all unique name axioms from $T_{\Delta,\Sigma}$.
Since neither \i{DCA} nor $U$ contains intensional predicates,
$\sm_{\bf p}[\Pi_{\Delta,\Sigma}]$ is equivalent to
$\sm_{\bf p}[\Pi_\Delta]\land\i{DCA}\land U$.  By Lemma~\ref{lem1}, it
follows that $\sm_{\bf p}[\Pi_{\Delta,\Sigma}]$ is equivalent to
the conjunction of the clauses~$\Delta$, the completion axioms~(\ref{comp}),
and the formulas \i{DCA} and $U$; that is to say, it is equivalent
to~$T_{\Delta,\Sigma}$.
\end{proof}

\subsection{Proof of Theorem~\ref{thm3}}

\medskip\noindent{\bf Theorem~\ref{thm3}.} {\it
If a subset~$X$ of~$\i{EHB}_\sigma$ is closed under substitution, and the set of
equalities in~$X$ is closed under reflexivity, symmetry, and transitivity,
then there exists a \i{DCA}-interpretation~$I$ such that $D(I)=X$.
Furthermore, this interpretation is unique up to isomorphism.
}

\medskip\begin{proof}
The binary relation
\beq
a=b\;\hbox{ is in }\;X
\eeq{ft1}
between object constants~$a$,~$b$ is an equivalence relation on the set of
object constants.  For any predicate constant~$P$, the $n$-ary relation
\beq
P(a_1,\dots,a_n)\;\hbox{ is in }\;X
\eeq{ft2}
between object constants~$a_1,\dots,a_n$ can be extended to equivalence
classes of~(\ref{ft1}).  Consider the interpretation~$I$ such that
\begin{itemize}
\item
the universe of~$I$ is the set of equivalence classes of relation~(\ref{ft1}),
\item
$I$ interprets each object constant~$a$ as the equivalence class that
contains~$a$,
\item
$I$ interprets each predicate constant~$P$ as the extension of the corresponding
relation~(\ref{ft2}) to equivalence classes.
\end{itemize}
Interpretation~$I$ satisfies \i{DCA}, and~$D(I)=X$.

To prove the second claim, consider any \i{DCA}-interpretation~$J$
such that $D(J)=X$.  For any object constant~$a$, let~$f(a)$ be the
element of the universe of~$J$ that represents~$a$.  Function~$f$ can be
extended to equivalence classes of relation~(\ref{ft1}), and this
extension is an isomorphism between~$I$ and~$J$.
\end{proof}

\subsection{Proof of Theorem~\ref{thm4}}

The proof of Theorem~\ref{thm4} is based on the fact that a
\i{DCA}-interpretation~$I$ satisfies a first-order sentence~$F$ of the
signature~$\sigma$ iff the
Herbrand interpretation $D(I)^{=}_{\eqs}$ satisfies $F^{=}_{\eqs}$.  This is easy to
verify by induction on the size of~$F$.  What we need actually is a
similar proposition for second-order sentences, because the
formulas obtained by applying the operator $\sm_{\bf p}$ contain predicate
variables.  The straightforward generalization to second-order sentences is
invalid, however.  For instance, let~$F$ be the formula
\beq
\exists v(v(a)\land\neg v(b))
\eeq{neq}
($v$ is a unary predicate variable).
This formula is equivalent to $a\neq b$.  If the universe of an
interpretation~$I$ is a singleton then~$I$ does not satisfy~$F$.  On the other
hand, the result of replacing~$=$ with~$\eq$ in~$F$ is~$F$ itself,
because this formula does not contain equality.  It is satisfied by every
Herbrand interpretation, including~$D(I)^{=}_{\eqs}$.

To overcome this difficulty, we will define the transformation $F\mapsto F^{=}_{\eqs}$
for second-order sentences in such a way that it will involve, in addition
to replacing~$=$ with~$\eq$, restricting the second-order quantifiers in~$F$.

In this section, a {\sl second-order formula} is a formula that may involve
predicate variables, either free or existentially quantified, but not function
variables.  (An extension to universally quantified predicate variables is
straightforward, but it is not needed for our purposes.)  For any predicate
variable~$v$, $\i{Sub}(v)$ stands for the formula
$$\forall x_1\cdots x_ny_1\cdots y_n(v(x_1,\dots,x_n)
\land \eq(x_1,y_1)\land\cdots\land\eq(x_n,y_n)\rar v(y_1,\dots,y_n)),$$
where~$n$ is the arity of~$v$.
For any second-order formula~$F$ of the signature~$\sigma$, $F^{=}_{\eqs}$ stands
for the second-order formula of the signature $\sigma\cup\{\eq\}$ obtained
from~$F$ by
\begin{itemize}
\item
replacing each  subformula of the form $t_1=t_2$ with $\eq(t_1,t_2)$, and
\item
restricting each second-order quantifier $\exists v$ to $\i{Sub}(v)$.
\end{itemize}
For instance, is~$F$ is~(\ref{neq}) then $F^{=}_{\eqs}$ is
$$\exists v(\i{Sub}(v)\land v(a)\land\neg v(b)).$$
In application to first-order formulas, the notation $F^{=}_{\eqs}$ has the same
meaning as before.

\begin{lemma}\label{lem2}
A \i{DCA}-interpretation~$I$ satisfies a second-order sentence~$F$  of the
signature~$\sigma$  iff the Herbrand interpretation $D(I)^{=}_{\eqs}$ satisfies
$F^{=}_{\eqs}$.
\end{lemma}

The proof of Lemma~\ref{lem2} is given in the online appendix.

\medskip
In the following lemma, as in the statement of Theorem~\ref{thm4},~$F$ is an
arbitrary sentence of the signature~$\sigma$, and~{\bf p} stands for the set
of all predicate constants of~$\sigma$.

\begin{lemma}\label{lem3}
For any \i{DCA}-interpretation~$I$ of the signature~$\sigma$,
$$I\models\sm_{\bf p}[F]\quad\hbox{iff}\quad D(I)^{=}_{\eqs}\models
\sm_{\bf p}[(F\land E_{\sigma})^{=}_{\eqs}].
$$
\end{lemma}

\begin{proof}
Recall that $\sm_{\bf p}[F]$ is defined as
$$
   F \land \neg \exists {\bf v} (({\bf v}<{\bf p}) \land F^*({\bf v}))
$$
\cite[Section~2.3]{fer09}, so that $\sm_{\bf p}[(F\land E_{\sigma})^{=}_{\eqs}]$ is
\beq
 F^{=}_{\eqs}\land (E_{\sigma})^{=}_{\eqs} \land\neg\exists {\bf v}(({\bf v}<{\bf p})
     \land F^*({\bf v})^{=}_{\eqs}\land  E_{\sigma}^*({\bf v})^{=}_{\eqs}).
\eeq{lem3a}
From the definitions of $ E_{\sigma}$ and of the
transformation $F\mapsto F^*({\bf v})$ \cite[Section~2.3]{fer09} we see that
$ E_{\sigma}^*({\bf v})$ is the conjunction of $ E_{\sigma}$ and the formulas
$$
\forall {\bf xy}(v({\bf x})\land {\bf x}={\bf y}\rar v({\bf y}))
$$
for all members~$v$ of tuple~{\bf v}.  Consequently
$ E_{\sigma}^*({\bf v})^{=}_{\eqs}$ is the conjunction of $ (E_{\sigma})^{=}_{\eqs}$ and the formulas
$\i{Sub}(v)$ for all members~$v$ of tuple {\bf v}.  It follows
that~(\ref{lem3a}) can be written as
$$
 F^{=}_{\eqs}\land (E_{\sigma})^{=}_{\eqs}
\land\neg\exists {\bf v}\left(({\bf v}<{\bf p})
   \land F^*({\bf v})^{=}_{\eqs}\land  (E_{\sigma})^{=}_{\eqs} \land \bigwedge_v\i{Sub}(v)\right).
$$
This formula is equivalent to
$$
 F^{=}_{\eqs} \land\neg\exists {\bf v}\left(\bigwedge_v\i{Sub}(v)
\land (({\bf v}<{\bf p}) \land F^*({\bf v}))^{=}_{\eqs}\right) \land (E_{\sigma})^{=}_{\eqs},
$$
which can be written as
$$\sm_{\bf p}[F]^{=}_{\eqs}\land (E_{\sigma})^{=}_{\eqs}.$$
The interpretation~$D(I)^{=}_{\eqs}$ satisfies the second conjunctive term.  By
Lemma~\ref{lem2},~$D(I)^{=}_{\eqs}$ satisfies the first conjunctive term
iff~$I$ satisfies~$\sm_{\bf p}[F]$.
\end{proof}

\medskip\noindent{\bf Theorem \ref{thm4}.} {\it
For any \i{DCA}-interpretation~$I$ of the signature~$\sigma$ that satisfies
$\sm_{\bf p}[F]$, the Herbrand interpretation $D(I)^{=}_{\eqs}$ of the
signature~$\sigma\cup\{\eq\}$ satisfies
$$
\sm_{\bf p}[(F\land E_{\sigma})^{=}_{\eqs}].
$$
Conversely, any Herbrand model of this formula is $D(I)^{=}_{\eqs}$ for some
\i{DCA}-interpreta\-tion~$I$ of~$\sigma$ satisfying $\sm_{\bf p}[F]$.
}

\medskip\begin{proof}
The first assertion is identical to the only-if part of Lemma~\ref{lem3}.
To prove the second assertion, consider a Herbrand model~$J$ of
$\sm_{\bf p}[(F\land E_{\sigma})^{=}_{\eqs}]$.  Since this formula entails
$ (E_{\sigma})^{=}_{\eqs}$,~$J$ is a model of~$ (E_{\sigma})^{=}_{\eqs}$ as well.
It follows that
the subset~$X$ of $\i{EDB}_\sigma$ such that $X^{=}_{\eqs}=J$ is closed under
substitution, and the set of equalities in~$X$ is closed under reflexivity,
symmetry, and transitivity.  By Theorem~\ref{thm3}, there exists a
\i{DCA}-interpretation~$I$ such that $D(I)=X$, so that $D(I)^{=}_{\eqs}=J$.
By the if part of Lemma~\ref{lem3},~$I$ satisfies $\sm_{\bf p}[F]$.
\end{proof}

\section{Conclusion}

This paper contributes to the direction of research on the semantics of null
values started in \cite{rei84} and \cite{tra94}.
More recently, null values were studied in the framework of the
Datalog$+/-$ project \cite{got10}.

We have demonstrated a close
relationship between Reiter's semantics of disjunctive databases and
cardinality constraints in answer set programming.  It shows also how answer
set solvers can be used for computing models of relational theories with
null values.

On the other hand, this paper improves our understanding of the role of
non-Herbrand stable models.  Are they
merely a mathematical curiosity, or can they have serious applications to
knowledge representation?  We have provided arguments in favor of the
usefulness of this generalization of the stable model semantics by showing,
first, how non-Herbrand stable models can serve for representing null values,
and second, how they can be generated using existing software systems.

The generalization of the stable model semantics proposed in \cite{fer09}
extends the original definition of a stable model in two ways: syntactically
(it is applicable to arbitrary first-order formulas) and semantically (a
stable model can be non-Herbrand).  The preprocessor {\sc f2lp} \cite{lee09a}
allows us to use existing answer set solvers for generating stable models
of some syntactically complex formulas.  On the other hand, the preprocessor
{\sc nonH.py}, described in this paper, allows us to use answer set solvers
for generating some non-Herbrand stable models---those that satisfy the
domain closure assumption but not the unique name assumption.  The two
programs can be used together.  For instance, the stable \i{DCA}-models of the
formula
$$(P(a) \land P(b)) \lor (P(c) \land P(d))$$
(there are 23 of them) can be generated by running {\sc f2lp} on the file
\begin{verbatim}
(p(a) & p(b)) | (p(c) & p(d)).
\end{verbatim}
and then running consecutively {\sc nonH.py}, {\sc gringo}, and {\sc claspD}.

\section*{Acknowledgements}

Many thanks to Marc Denecker, Michael Gelfond, and Yuliya Lierler for
useful discussions related to the topic of this paper, and to the
anonymous referees for valuable advice.

\bibliographystyle{acmtrans}
\bibliography{bib}
\end{document}


\maketitle

\noindent{\it Lemma 2}

\noindent
A \i{DCA}-interpretation~$I$ satisfies a second-order sentence~$F$  of the
signature~$\sigma$  iff the Herbrand interpretation $D(I)^{=}_{\eqs}$ satisfies
$F^{=}_{\eqs}$.

\begin{proof}
The proof is by induction on the size of~$F$;
size is understood as follows. About second-order sentences~$F$ and~$G$ we say
that~$F$ is {\sl smaller than}~$G$ if
\begin{itemize}
\item
$F$ has fewer second-order quantifiers than~$G$, or
\item
$F$ has the same number of second-order quantifiers as~$G$, and the
total number of first-order quantifiers and propositional connectives in~$F$
is less than in~$G$.
\end{itemize}
The induction hypothesis is that the assertion of the lemma holds for all
sentences that are smaller than~$F$. If $F$ is atomic then
$$\ba {ll}
I\models F\;   &\hbox{iff}\quad F\in D(I)\\
               &\hbox{iff}\quad F^{=}_{\eqs}\in D(I)^{=}_{\eqs}\\
               &\hbox{iff}\quad D(I)^{=}_{\eqs}\models F^{=}_{\eqs}.
\ea
$$
If $F$ is $G\land H$ then $F^{=}_{\eqs}$ is $G^{=}_{\eqs}\land H^{=}_{\eqs}$.  Using the induction
hypothesis, we calculate:
$$\ba{ll}
I\models F\; &\hbox{iff}\quad I\models G\hbox{ and }I\models H\\
             &\hbox{iff}\quad D(I)^{=}_{\eqs}\models G^{=}_{\eqs} \hbox{ and }
                              D(I)^{=}_{\eqs}\models H^{=}_{\eqs}\\
             &\hbox{iff}\quad D(I)^{=}_{\eqs}\models F^{=}_{\eqs}.
\ea
$$
For other propositional connectives the reasoning is similar.
If $F$ is $\forall xG(x)$ then $F^{=}_{\eqs}$ is $\forall x\left(G(x)^{=}_{\eqs}\right)$.
Using the induction hypothesis and the fact that~$I$ satisfies \i{DCA},
we calculate:
$$\ba{ll}
I\models F\;
 &\hbox{iff}\quad\hbox{for all object constants}\ a,\ I\models G(a)\\
 &\hbox{iff}\quad\hbox{for all object constants}\ a,\ D(I)^{=}_{\eqs}\models G(a)^{=}_{\eqs}\\
 &\hbox{iff}\quad D(I)^{=}_{\eqs}\models F^{=}_{\eqs}.
\ea
$$
For the first-order existential quantifier the reasoning is similar.

It remains to consider the case when $F$ is $\exists vG(v)$, where~$v$ is
a predicate variable.
To simplify notation, we will assume that the arity of~$v$ is~1.
For any set~$V$ of object constants, by
$\i{exp}_V$ we denote the lambda-expression\footnote{On the use of
lambda-expressions in logical formulas, see \cite[Section~3.1]{lif93e}.}
$\lambda x\bigvee_{a\in V}(x=a)$.  Since~$I$ is a \i{DCA}-interpretation,
$I\models F$ iff
$$\hbox{for some }V,\quad I\models G(\i{exp}_V).$$
By the induction hypothesis, this is equivalent to the condition
\beq
\hbox{for some }V,\quad
   D(I)^{=}_{\eqs}\models H((\i{exp}_V)^{=}_{\eqs}),
\eeq{lhs}
where $H(v)$ stands for $G(v)^{=}_{\eqs}$.
On the other hand, $F^{=}_{\eqs}$ is $\exists v(\i{Sub}(v)\land H(v))$.  The
Herbrand interpretation $D(I)^{=}_{\eqs}$ satisfies this formula iff
\beq
\hbox{for some }V,\quad
   D(I)^{=}_{\eqs}\models\i{Sub}(\i{exp}_V)
   \;\hbox{ and }\; D(I)^{=}_{\eqs}\models H(\i{exp}_V).
\eeq{rhs}
We need to show that~(\ref{rhs}) is equivalent to~(\ref{lhs}).

Consider first the part
\beq
D(I)^{=}_{\eqs}\models\i{Sub}(\i{exp}_V)
\eeq{rhs1}
of condition~(\ref{rhs}), that is,
$$D(I)^{=}_{\eqs}\models\forall xy(\i{exp}_V(x)\land \eq(x,y)\rar \i{exp}_V(y)).$$
It is equivalent to
$$
D(I)^{=}_{\eqs}\models\forall y(\exists x(\i{exp}_V(x)\land \eq(x,y))\rar \i{exp}_V(y)).
$$
Interpretation $D(I)^{=}_{\eqs}$ satisfies the inverse of this implication,
because it satisfies $\forall x\eq(x,x)$.  Consequently condition~(\ref{rhs1})
can be equivalently rewritten as
$$
D(I)^{=}_{\eqs}\models\forall y(\exists x(\i{exp}_V(x)\land \eq(x,y))\lrar \i{exp}_V(y)).
$$
The left-hand side of this equivalence can be rewritten as
$\bigvee_{a\in V}\eq(a,y)$.  It follows that condition~(\ref{rhs1}) is equivalent
to
$$
\ba c
D(I)^{=}_{\eqs}\models\forall y\left(\bigvee_{a\in V}\eq(a,y)\lrar \i{exp}_V(y)\right).
\ea
$$
Furthermore, $\eq(a,y)$ can be replaced here by $\eq(y,a)$, because
$D(I)^{=}_{\eqs}$ satisfies $\forall xy(\eq(x,y)\lrar \eq(y,x))$.  Hence~(\ref{rhs1}) is
equivalent to
$$
\ba c
D(I)^{=}_{\eqs}\models(\i{exp}_V)^{=}_{\eqs}=\i{exp}_V.
\ea
$$
It follows that~(\ref{rhs}) is equivalent to the condition
\beq
\hbox{for some }V,\,
   D(I)^{=}_{\eqs}\models(\i{exp}_V)^{=}_{\eqs}=\i{exp}_V\hbox{ and }
   D(I)^{=}_{\eqs}\models H((\i{exp}_V)^{=}_{\eqs}).
\eeq{rhs2}
It is clear that (\ref{rhs2}) implies (\ref{lhs}).

It remains to check that (\ref{lhs}) implies (\ref{rhs2}).
Assume that
\beq
D(I)^{=}_{\eqs}\models H((\i{exp}_V)^{=}_{\eqs}),
\eeq{prime}
and let~$V'$ be the set of object constants~$a$ such that, for some~$b\in V$,
$I\models a=b$.  We will show that $V'$ can be taken as~$V$ in~(\ref{rhs2}).
The argument uses two properties of the set~$V'$ that are immediate from its
definition:
\begin{enumerate}
\item[(a)] $V\subseteq V'$;
\item[(b)] if $I\models a=b$ and $a\in V'$ then $b\in V'$.
\end{enumerate}

Consider the first half of~(\ref{rhs2}) with $V'$ as~$V$:
$$
D(I)^{=}_{\eqs}\models(\i{exp}_{V'})^{=}_{\eqs}=\i{exp}_{V'}.
$$
This condition can be restated as follows: for every object constant~$a$,
$$
\ba c
D(I)^{=}_{\eqs}\models\bigvee_{b\in V'}\eq(a,b)\quad\hbox{iff}\quad
D(I)^{=}_{\eqs}\models\bigvee_{b\in V'}(a=b),
\ea
$$
or, equivalently,
$$
\ba c
I\models\bigvee_{b\in V'}(a=b)\quad\hbox{iff}\quad a\in V'.
\ea
$$
The implication left-to-right follows from property~(b) of $V'$; the
implication right-to-left is obvious (take~$b$ to be~$a$).

Consider now the second half of~(\ref{rhs2}) with $V'$ as~$V$:
$$
   D(I)^{=}_{\eqs}\models H((\i{exp}_{V'})^{=}_{\eqs}).
$$
To derive it from~(\ref{prime}), we only need to check that
$$
D(I)^{=}_{\eqs}\models (\i{exp}_{V'})^{=}_{\eqs}=(\i{exp}_{V})^{=}_{\eqs}.
$$
This claim is equivalent to
\beq
I\models \i{exp}_{V'}=\i{exp}_{V}
\eeq{end}
and can be restated as follows: for every object constant~$a$,
$$
\ba c
I\models\bigvee_{b\in V'}(a=b)\quad\hbox{iff}\quad
I\models\bigvee_{b\in V}(a=b).
\ea
$$
The implication left-to-right is immediate from the definition of~$V'$; the
implication righ-to-left is immediate from property~(a).
\end{proof}

\bibliographystyle{acmtrans}
\bibliography{bib}